\documentclass[journal]{IEEEtran}

\usepackage[latin1]{inputenc}
\usepackage{graphicx,graphics,subfigure,xspace}
\usepackage{amssymb,amsmath,amsfonts,amsthm}
\usepackage{algorithm,algorithmic}
\usepackage{cite,float}
\usepackage[usenames,dvipsnames]{color}
\usepackage{epsfig, hyperref}

\newtheorem{lemma}{Lemma}
\newtheorem{prop}{Proposition}

\DeclareMathOperator{\Tr}{Tr}
\DeclareMathOperator{\Rk}{rank}
\newcommand{\papr}{{\rm PAPR}}

\begin{document}

\title{Tone-index Multisine Modulation for SWIPT}

\author{Ioannis Krikidis,~\IEEEmembership{Fellow,~IEEE}, and Constantinos Psomas,~\IEEEmembership{Senior Member,~IEEE}\vspace{-0.7cm}
\thanks{I. Krikidis and C. Psomas are with the Department of Electrical and Computer Engineering, University of Cyprus, Nicosia 1678 (E-mail: {\sf\{krikidis, psomas\}@ucy.ac.cy}). This work has received funding from the European Research Council (ERC) under the European Union's Horizon 2020 research and innovation programme (Grant agreement No. 819819).}}

\maketitle

\begin{abstract}
We propose a new simultaneous wireless information and power transfer (SWIPT) technique that embeds information bits in the tone-index of multisine waveforms. By varying the number of subcarriers of the transmitted bandwidth-constrained multisine signal, the proposed scheme enables efficient radio-frequency energy harvesting and low-complexity information transmission. The receiver does not require channel estimation and employs a non-coherent maximum-likelihood detection at the envelope of the received signal. The performance of the proposed tone-index modulation is evaluated in terms of average error probability for a flat-fading channel, and we show that it outperforms its peak-to-average-power-ratio counterpart.
\end{abstract}
\vspace{-0.1cm}

\vspace{-0.2cm}
\section{Introduction}

\IEEEPARstart{S}{imultaneous} wireless information and power transfer (SWIPT) is a new communication paradigm, where information and energy flows are co-designed/co-engineered to concurrently communicate and energize \cite{KRI,CLE}. It is an enabling technology to provide ubiquitous connectivity and energy sustainability in low-power wireless devices. The implementation of SWIPT requires the design of new waveforms and low-complexity receiver architectures that increase the harvesting efficiency \cite{CLA1,CLA2}. It has been  shown that multisine signals boost the energy harvesting by exploiting the non-linearity of the rectification process \cite{CLER}.

 Orthogonal SWIPT architectures (i.e., time switching, power splitting) split the received signal in two distinct parts, one for information decoding and one for energy harvesting \cite{KRI, GAO}. Despite their efficiency to achieve various information-energy tradeoffs, these approaches require radio-frequency (RF) to baseband downconversion and correspond to high power consumption, which is not inline with the energy harvesting limitations that characterizes SWIPT. The first technique that overcomes this drawback is the integrated receiver, where information is embedded into the levels of energy signals and decoding is achieved by employing coherent detection at the output of the rectification circuit \cite{RUI1}. A more recent solution does not require channel estimation and utilizes multi-sine waveforms of variable bandwidth for energy transfer, while their distinct peak-to-average-power-ratios (PAPRs) convey information \cite{KIM}. Although this new technique is suitable for low-power devices, PAPR-based information decoding works only for high signal-to-noise ratios (SNRs) and therefore cannot be used for low and/or moderate SNRs. 

In this letter, we investigate a new SWIPT technique that exploits bandwidth-constrained multisine waveforms for RF energy harvesting, while information is embedded in the number of tones. By assuming only a statistical channel state information, we employ a (low-complexity) non-coherent maximum-likelihood (ML) detection on the received signal. We characterize the performance of the proposed scheme in terms of the average error probability for a flat-fading channel, based on the union bound over the exact expressions of the pairwise error probabilities. The proposed tone-index multisine (TIM) modulation outperforms its PAPR counterpart \cite{KIM} and ensures connectivity for a larger SNR regime. The TIM modulation and the associated decoding scheme are promising designs for practical SWIPT implementations.

\vspace{-0.2cm}
\section{System Model}
We consider a simple point-to-point topology consisting of one transmitter, aiming to simultaneously convey information and energy to a single receiver. The transmitter generates an unmodulated $N$-tone multisine signal of bandwidth $W=(N-1)\Delta f^{(N)}$ with zero phase arrangement and intercarrier frequency spacing $\Delta f^{(N)}$ \cite{PAN}. In comparison to continuous wave signals with the same average power, multisine waveforms provide higher PAPR and can more easily overcome the build-in potential of the rectifying devices \cite{CLE,PAN}. The baseband equivalent is given by 
\begin{align}
x(t)=\sqrt{\frac{2P}{N}} \Re \left\{ \frac{\sin(\pi N \Delta f^{(N)} t)}{\sin(\pi \Delta f^{(N)} t)}e^{j2\pi f_c t} \right\}, 0 \leq t \leq T,
\end{align}
where $\Re\{z\}$ denotes the real part of $z$, $f_c$ is the carrier frequency, $P$ is the total transmit power and $T$ is the symbol time. It is worth noting that the transmitted signal corresponds to a conventional amplitude modulated (AM) signal. The transmitted AM signal propagates through a frequency-flat block fading\footnote{An appropriate system design ensures operation over flat-fading channels; in this case, uniform power allocation across the tones becomes optimal \cite{CLE2}.} wireless channel $h(t)$, and thus the received signal is written as
\begin{align}
r(t)=\sqrt{\frac{2P}{N}} \Re \left\{ \frac{\sin(\pi N \Delta f^{(N)} t)}{\sin(\pi \Delta f^{(N)} t)}h(t)e^{j2\pi f_c t} \right\}+n(t),
\end{align}
where $n(t)$ denotes the additive white Gaussian noise (AWGN) component.
The received signal is driven to a rectenna circuit for RF-to-DC conversion and power management. For the information transfer, the receiver employs an AM linear envelope detector on the received signal (similar to \cite{KIM}), which follows a tone-index demodulator block; diode-based distortion is not considered for the sake of simplicity. By assuming baseband representation, the discrete-time multisine signal envelope is
\begin{align}
r[k]=h x^{(N)}[k]+n[k],\;\;\text{with}\; 1\leq k\leq K, \label{eq1}
\end{align}
where $x^{(N)}[k]=\sqrt{\frac{P}{N}} \sin[\pi N\Delta f^{(N)} k]/\sin[\pi \Delta f^{(N)} k]$ denotes the $k$-th sample of the transmitted waveform, $h$ is the (real) channel coefficient which is considered constant for a symbol time (so we omit the time dependency), $n[k]$ is the $k$-th sample of the AWGN, and $K$ is the total number of samples per symbol. We assume $h\sim \mathcal{N}(0,\sigma^2_h)$ and $n\sim \mathcal{N}(0,\sigma^2_n)$, where $\mathcal{N}({\mu,\sigma^2})$ indicates the Gaussian distribution with mean $\mu$ and variance $\sigma^2$.

\noindent 
{\it RF energy harvesting:} By considering the non-linear physics-based diode model, the DC energy harvested associated with an $N$-tone multisine signal is given by \cite{CLE2}
\begin{align}
Q_N=\!\!\sum_{i=2,4}a_i\mathbb{E}\{|r(t)|^i\} =a_2\sigma_h^2 P+ 3 a_4 \sigma_h^4 \frac{2N^2+1}{2N} P^2, \label{EH}
\end{align}
where $a_2$, $a_4$ are constants related to the characteristics of the rectification circuit. 
 
 
\begin{figure}\centering
  \includegraphics[width=0.93\linewidth]{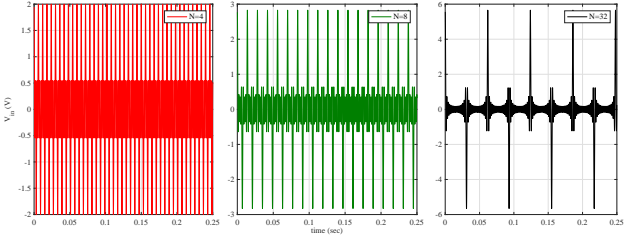}\vspace{-1mm}
  \caption{Time-domain waveform of a $N$-tone signal with  $N\in \mathcal{S}=\{4,8,32\}$.}\label{fig1}\vspace{-1.5mm}
\end{figure}

\vspace{-0.2cm}
\section{Tone-index multisine modulation}

The proposed TIM scheme embeds information into the number of tones of the transmitted multisine signal. At each symbol time, the transmitter conveys information to the receiver by varying the number of subcarriers $N$, based on a predefined set $\mathcal{S}$ of multisine waveforms. As an example, Fig. \ref{fig1} presents a set of three multisine signals with $N\in \mathcal{S}=\{4,8,32\}$; by selecting one waveform at each symbol time, the transmitter conveys $\log_2(3)$ bits of information. To extract the information, the receiver employs a low-complexity envelope detector which converts the received high-frequency AM signal to the envelope of the original multisine signal. The output of the envelope detector feeds a tone-index demodulation which estimates the number of subcarriers by employing a non-coherent ML detection. In this way, the transmitter continuously transmits energy-efficient multisine waveforms of bandwidth $W$, and information transfer is performed without power-hungry RF to baseband downconversion. Although the proposed architecture requires an oversampling of $K$ samples/symbol, simulation results demonstrate that small values of $K$ are sufficient for successful decoding. Given that the system bandwidth is low, the cost associated to the analog-to-digital conversion process is negligible \cite{SCO}. Fig. \ref{fig2} schematically presents the proposed SWIPT architecture.

The proposed TIM scheme transmits a fixed set of multisine waveforms and maintains a constant bandwidth $W$ by adapting $\Delta f^{(N)}$ accordingly; on the other hand, the PAPR-based scheme in \cite{KIM} corresponds to a variable bandwidth. In addition, a one-to-one relationship between the transmitted waveforms and the PAPRs does not exist since the transmitted tones in \cite{KIM} vary based on the available channel state information. 

\begin{figure}
  \includegraphics[width=0.9\linewidth]{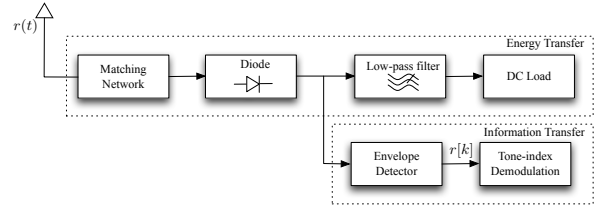}\vspace{-1mm}
  \caption{A new SWIPT architecture based on TIM modulation.}\label{fig2}\vspace{-1.5mm}
\end{figure}

\vspace{-0.2cm}
\subsection{ML detection and average error probability}

The index-tone demodulator employs an ML non-coherent detector based on the linear expression in \eqref{eq1}. Specifically, conditioned on the transmitted waveform $\mathbf{x}^{(i)}=[{x}^{(i)}[1], {x}^{(i)}[2],\cdots, {x}^{(i)}[K] ]^T$, the received vector $\mathbf{r}$ is Gaussian with probability density function (PDF)
\begin{align}
p(\mathbf{r}| \mathbf{x}^{(i)})=\frac{1}{\sqrt{|2\pi \mathbf{R}_i|}}e^{-\frac{1}{2}\mathbf{r}^T\mathbf{R}_i^{-1}\mathbf{r}},
\end{align}
where $\mathbf{R}_i=\sigma_n^2\mathbf{I}+\sigma_h^2\mathbf{x}^{(i)}(\mathbf{x}^{(i)})^T$ denotes the covariance matrix; the ML detector (log-likelihood) can be summarized as
\begin{align}
\hat{N}=\arg\max_{i\in \mathcal{S}} p(\mathbf{r}| \mathbf{x}^{(i)})=\arg\max_{i\in \mathcal{S}} \log p(\mathbf{r}| \mathbf{x}^{(i)}).
\end{align}
To assess the performance of the ML detector, we use the pairwise error probability \cite{SKO}
\begin{align}
&\mathbb{P}\{\mathbf{x}^{(i)}\rightarrow \mathbf{x}^{(j)} \}=\mathbb{P}\left\{\log p(\mathbf{r}| \mathbf{x}^{(j)})>\log p(\mathbf{r}| \mathbf{x}^{(i)})\bigg|\mathbf{x}^{(i)} \right\} \nonumber \\
&=\mathbb{P}\left\{\mathbf{r}^T (\mathbf{R}_i^{-1}-\mathbf{R}_j^{-1})\mathbf{r} > \phi_{ij} \right\} \nonumber \\
&= \mathbb{P}\left\{\mathbb{\zeta}^T \mathbf{R}_i^{\frac{T}{2}} \left(\mathbf{R}_i^{-1}-\mathbf{R}_j^{-1}\right)\mathbf{R}_i^{\frac{1}{2}} \mathbb{\zeta}> \phi_{ij} \right\},\label{pep}
\end{align}
where $\phi_{ij}= \log \frac{|\mathbf{R}_j|}{|\mathbf{R}_i|}$ and the vector $\mathbf{\zeta} = \mathbf{R}_i^{-\frac{1}{2}} \mathbf{r}$ has elements $\zeta_i$, which are independent and identically distributed (i.i.d.) zero-mean, unit-variance real Gaussian random variables. To further analyze the pairwise probability, we state the following Lemma.

\begin{lemma}\label{lemma1}
The rank of the matrix $\mathbf{R}_i^{\frac{T}{2}} \left(\mathbf{R}_i^{-1}-\mathbf{R}_j^{-1}\right)\mathbf{R}_i^{\frac{1}{2}}$ is two with eigenvalues $\mu_1=\sigma_n^2 \lambda_1$ and $\mu_2=(\sigma_n^2+\sigma_h^2 (\mathbf{x}^{(i)})^T(\mathbf{x}^{(i)}))\lambda_2$, where $\lambda_1$, $\lambda_2$ are the two non-zero eigenvalues of the matrix $\mathbf{R}_i^{-1}-\mathbf{R}_j^{-1}$ with opposite signs.  
\end{lemma}

\begin{proof}
See Appendix \ref{lemma1_prf}.
\end{proof}

We can now state the following proposition.

\begin{prop}\label{prop1}
The pairwise error probability for the tone-index ML detection is
\begin{align}
\mathbb{P}\{\mathbf{x}^{(i)}\rightarrow \mathbf{x}^{(j)}\} = \frac{1}{\pi} \sum_{j=0}^\infty& \frac{-2\sqrt{-1}}{j!(2j+1)} \left(\frac{\mu_2}{\mu_1}\right)^{j+\frac{1}{2}} \Gamma\left(j+\frac{3}{2}\right)\nonumber\\
&\times e^{-\frac{z_0}{2}} U\left(\frac{1}{2},-j,\frac{z_0}{2}\right),
\end{align}
where $z_0 = \max(0,\phi_{ij}/\mu_2)$ and $U(\cdot,\cdot,\cdot)$ is the confluent hypergeometric function of the first kind.
\end{prop}

\begin{proof}
See Appendix \ref{prop1_prf}.
\end{proof}

Therefore, the union bound on the total error probability for TIM can be written as
\begin{align}\label{ptim}
P^{\rm TIM}_e \leq \frac{1}{|\mathcal{S}|} \sum_{i\in \mathcal{S}} \sum_{j\in\mathcal{S},j\neq i} \mathbb{P}\{\mathbf{x}^{(i)}\rightarrow \mathbf{x}^{(j)}\},
\end{align}
where $|\mathcal{S}|$ is the cardinality of the set $\mathcal{S}$. 

\subsection{PAPR-based SWIPT}

The PAPR-based SWIPT scheme has been proposed in \cite{KIM}, but its performance has been evaluated only for the high SNR regime. For comparison reasons, we study its performance for all SNR values and we provide closed-form expressions for the achievable error probability. More specifically, the PAPR for a multisine signal with $N$ tones is
\begin{align}
\papr(N) = \frac{\max_k |r[k]|^2}{\mathbb{E}\{|r[k]|^2\}},
\end{align}
where the expectation is taken over a symbol time. The following proposition provides an approximation for the cumulative distribution function (CDF) of the $N$-tone PAPR. 
\begin{prop}\label{prop2}
The CDF of PAPR is approximated by
\begin{align}
F_{\rm PAPR}&(\theta,N) = \prod_k \Bigg(1-\frac{1}{\sqrt{2\pi\sigma_h^2}} \int_{-\infty}^\infty e^{-\frac{z^2}{2\sigma_h^2}}\nonumber\\
&\times Q_\frac{1}{2} \left(\frac{z x^{(N)}[k]}{\sigma_n},\theta^\frac{1}{2} \left(\frac{z^2\xi(N)}{\sigma^2_n} + 1\right)^\frac{1}{2}\right) dz\Bigg),
\end{align}
where $\xi(N) \triangleq \frac{P}{NK}\sum_{k=1}^K\frac{\sin^2[\pi N\Delta f^{(N)} k]}{\sin^2[\pi \Delta f^{(N)} k]}$ and $Q_\frac{1}{2}(\cdot,\cdot)$ is the Marcum $Q$-function of order $1/2$.
\end{prop}

\begin{proof}
See Appendix \ref{prop2_prf}.
\end{proof}

\begin{figure}[t]\centering
  \includegraphics[width=0.8\linewidth]{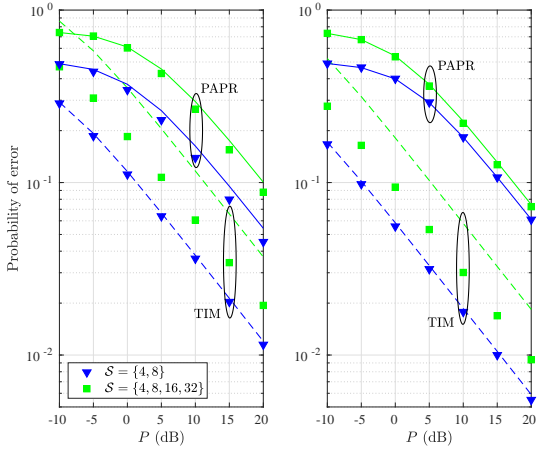}\vspace{-2mm}
  \caption{Probability of error for TIM and PAPR modulation; lines and markers depict theoretical and simulation results, respectively.}\label{fig3}\vspace{-1.5mm}
\end{figure}

Using the above CDF, we can deduce the error probability for the PAPR-based SWIPT. Specifically, the demodulator decides $\hat{N}$ based on the minimum one-dimensional Euclidean distance, i.e., $\hat{N}=\arg\min_{s_i\in \mathcal{S}} |\papr(N)-s_i|$, where $s_i$ is the $i$-th element of $\mathcal{S}$. Let $d_i = \frac{1}{2}(s_i+s_{i+1})$ be the midpoint between $s_i$ and $s_{i+1}$. The midpoints serve as the decision boundaries of the demodulator. Hence,
\begin{align}\label{ppapr}
P^{\rm PAPR}_e = \frac{1}{|\mathcal{S}|} \sum_{i\in \mathcal{S}} p_i,
\end{align}
where $p_1 = 1 - F_{\rm PAPR}(d_1)$, $p_{|\mathcal{S}|} = F_{\rm PAPR}(d_{|\mathcal{S}|-1})$ and $p_i = 1 - F_{\rm PAPR}(d_i) + F_{\rm PAPR}(d_{i-1})$, $1 < i < |\mathcal{S}|$.

\subsection{Information-energy tradeoff}\label{ieto}

The proposed SWIPT scheme is also associated with a fundamental information-energy tradeoff. Assume a symbol time $T \geq 1/\Delta f^*$, where $\Delta f^* = (N^*-1)/W$ is upper bounded by the minimum frequency spacing between tones, i.e., $N^* \geq \max \mathcal{S}$. Then, the information rate and the energy harvested are equal to
\begin{equation}
R=\frac{\log_2(|\mathcal{S}|)}{T}~ {\rm bits/sec},
\end{equation}
\begin{equation}\label{energy}
Q=\frac{1}{|\mathcal{S}|}\sum_{N \in \mathcal{S}}Q_N\approx a_2\sigma_h^2 P+\frac{3 a_4\sigma_h^4P^2}{|\mathcal{S}|}\sum_{N \in \mathcal{S}}N,
\end{equation}
respectively, where the approximation in \eqref{energy} holds for large $N$. The achieved information-energy trade-off is similar to the PAPR-based scheme; however, the proposed TIM scheme ensures a larger operation SNR regime and a lower average error probability.

\begin{figure}[t]\centering
  \includegraphics[width=0.8\linewidth]{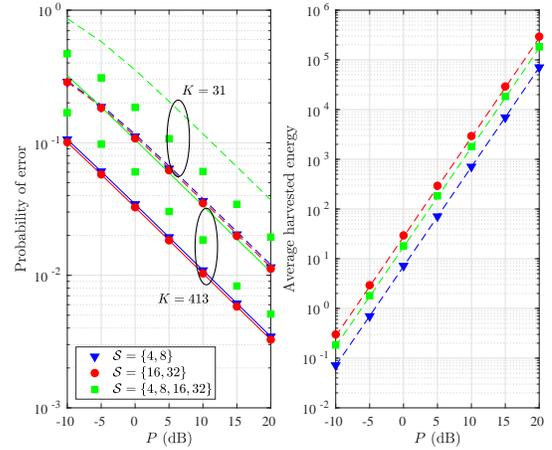}\vspace{-2mm}
  \caption{Probability of error and average harvested energy using TIM; lines and markers depict theoretical and simulation results, respectively.}\label{fig4}\vspace{-1.5mm}
\end{figure}

\section{Numerical results \& conclusions}\label{concl}
The following parameters were used for the simulations: $\sigma_h^2 = 1$, $\sigma_n^2 = 1$, $W = 1$ kHz, $a_2=0.0034$ and $a_4=0.3829$.

Fig. \ref{fig3} illustrates the probability of error for both TIM and PAPR modulation in terms of the transmit power, with $\mathcal{S} = \{4,8\}$ and $\mathcal{S} = \{4,8,16,32\}$. The symbol time is $T = 1/\Delta f^* = (N^*-1)/W$, with $N^* = 32$, $K = 31$ samples/symbol and $N^* = 128$, $K = 127$ samples/symbol, for the left and right sub-figure, respectively. It is clear that for both symbol times, our proposed TIM modulation outperforms the PAPR approach; higher gains are achieved with a larger symbol time. Moreover, a smaller set of multisine waveforms provides a better probability of error, as expected. Finally, the simulation results (markers) are in line with our theoretical analysis (lines). For TIM, expression \eqref{ptim} provides an upper bound when $|\mathcal{S}| > 2$, whereas it matches the simulation when $|\mathcal{S}| = 2$. For PAPR, expression \eqref{ppapr} provides a good approximation and converges to the simulation results as $T$ increases.

Fig. \ref{fig4} depicts the error probability and the average harvested energy attained by the TIM scheme, with $N^* = 32$ and for $K=31$ samples/symbol (dashed lines) and $K=413$ samples/symbol (solid lines). The main observation is that the set of multisine waveforms $\{16,32\}$ achieves the best performance with respect to both the error probability and the average harvested energy. On the other hand, the set $\{4,8,16,32\}$ achieves the worst error probability performance and comes second best in terms of average harvested energy. As expected, a larger $K$ provides a lower probability of error but has no effect on the average harvested energy. However, as discussed in Section \ref{ieto}, there exists an information-energy tradeoff. Specifically, $\{16,32\}$ provides an information rate equal to $1/T$ whereas $\{4,8,16,32\}$ an information rate equal to $2/T$. Therefore, these insights provide design guidelines for SWIPT systems employing TIM, according to the rate/energy requirements. The practical application of the TIM scheme for flat-fading channels requires a careful system design (e.g. bandwidth, number of tones, inter-tone spacing etc). An investigation of the system parameters on the achieved performance as well as the application of the TIM scheme to frequency selective channels, are interesting future research directions.

\appendices
\section{Proof of Lemma \ref{lemma1}}\label{lemma1_prf}
Since the matrix $\mathbf{R}_i=\sigma_n^2\mathbf{I}+\sigma_h^2 \mathbf{x}^{(i)}(\mathbf{x}^{(i)})^T$ is the sum of a diagonal matrix with a rank-$1$ matrix, by employing the Sherman-Morrison formula \cite{HOR}, we have $\mathbf{R}_i^{-1}=(1/\sigma_n^2)\mathbf{I}-\mathbf{Q}_i$, where $\mathbf{Q}_i=(\sigma_h^2/\sigma_n^4)\mathbf{x}^{(i)}(\mathbf{x}^{(i)})^T/[1+(\sigma_h^2/\sigma_n^2)(\mathbf{x}^{(i)})^T(\mathbf{x}^{(i)})]$ is a rank-$1$ matrix. The matrix
$\mathbf{W}_{ij}=\mathbf{R}_i^{-1}-\mathbf{R}_j^{-1}=\mathbf{Q}_j-\mathbf{Q}_i$ is the sum of two rank-$1$ matrices, and since $\mathbf{x}^{(i)}$ and $\mathbf{x}^{(j)}$ are linearly independent, it is a rank-$2$ matrix. By calculating the trace of the matrix $\mathbf{W}_{i,j}$, we can find an expression for its two non-zero eigenvalues $\lambda_1$, $\lambda_2$. Specifically, we have
\begin{align}
&\Tr({\mathbf{W}}_{ij})=\Tr({\mathbf{Q}_j})-\Tr({\mathbf{Q}_i}) \nonumber \\
&=\frac{\sigma_h^2}{\sigma_n^4}\frac{\Tr({\mathbf{x}^{(j)}(\mathbf{x}^{(j)})^T})}{1+\frac{\sigma_h^2}{\sigma_n^2}\Tr({\mathbf{x}^{(j)}(\mathbf{x}^{(j)})^T})}-\frac{\sigma_h^2}{\sigma_n^4}\frac{\Tr({\mathbf{x}^{(i)}(\mathbf{x}^{(i)})^T})}{1+\frac{\sigma_h^2}{\sigma_n^2}\Tr({\mathbf{x}^{(i)}(\mathbf{x}^{(i)})^T})} \nonumber\\
&\approx 0 \Rightarrow \lambda_1+\lambda_2=0,
\end{align} 
with\footnote{We have that $\frac{P}{i} \int_{0}^{T}\frac{\sin^2(i \pi \Delta f^{(i)} t)}{\sin^2(\pi \Delta f^{(i)} t)}dt\geq \frac{P}{i} \int_{0}^{\frac{k}{\Delta f^{(i)}}}\frac{\sin^2(i \pi \Delta f^{(i)} t)}{\sin^2(\pi \Delta f^{(i)} t)}dt=\frac{P}{i}\int_{0}^{k\pi} \frac{\sin^2 (ix)}{\sin^2 (x)} dx=\frac{Pk}{i}\int_{0}^{\pi}\frac{\sin^2 (ix)}{\sin^2 (x)} dx$ $=Pk \pi\gg 1$ \cite{GRAD} for moderate $k$ and $P$, where $T=\frac{k}{\Delta f^{(i)}}+\tau$ with $\tau \leq 1/\Delta f^{(i)}$ and $k$ positive integer.} $\Tr({\mathbf{x}^{(i)}(\mathbf{x}^{(i)})^T})\approx \frac{P}{i} \int_{0}^{T}\frac{\sin^2(i \pi \Delta f^{(i)} t)}{\sin^2(\pi \Delta f^{(i)} t)}dt\gg1\; \forall\; i$, which shows that the two non-zero eigenvalues have different signs and are equal in absolute value. Given that the matrices $\mathbf{R}_i$ and $\mathbf{W}_{ij}$ are symmetric, the rank of the matrix $\mathbf{R}_i^{\frac{T}{2}}\mathbf{W}_{ij}\mathbf{R}_i^{\frac{1}{2}}$ is deduced by employing the Sylvester and rank-sum inequalities \cite{HOR},
\begin{align}
&\Rk(\mathbf{W}_{ij})+\Rk(\mathbf{R}_i) - K = 2 \leq \Rk(\mathbf{R}_i^{\frac{T}{2}}\mathbf{W}_{ij}\mathbf{R}_i^{\frac{1}{2}}) \nonumber \\
&=\Rk(\mathbf{W}_{ij}\mathbf{R}_i) \leq \min(\Rk(\mathbf{W}_{ij}),\Rk(\mathbf{R}_i)) = 2\nonumber \\
& \Rightarrow \Rk(\mathbf{R}_i^{\frac{T}{2}}\mathbf{W}_{ij}\mathbf{R}_i^{\frac{1}{2}})=2,
\end{align}
where we note that the matrices $\mathbf{W}_{ij}\mathbf{R}_i$ and $\mathbf{R}_i^{\frac{T}{2}}\mathbf{W}_{ij}\mathbf{R}_i^{\frac{1}{2}}$ have the same eigenvalues. Since $\mathbf{W}_{ij}$ and $\mathbf{R}_i$ are diagonalizable (symmetric), there exists $\mathbf{P}$ such that 
$\mathbf{W}_{ij}=\mathbf{P}\mathbf{D}_w\mathbf{P}^{-1}$ and $\mathbf{R}_i=\mathbf{P}\mathbf{D}_R\mathbf{P}^{-1}$ and hence $\mathbf{W}_{ij}\mathbf{R}_i=\mathbf{P}\mathbf{D}_w\mathbf{D}_R\mathbf{P}^{-1}$; this means that the eigenvalues of $\mathbf{W}_{ij}\mathbf{R}_i$ are given by the product of the individual eigenvalues of the matrices $\mathbf{W}_{ij}$ and $\mathbf{R}_{i}$. The matrix $\mathbf{R}_i$ has $K$ eigenvalues with $\nu_k=\sigma_n^2$ with $k=1,\ldots,K-1$ (multiplicity $K-1$) and $\nu_K=\sigma_n^2+\sigma_h^2 (\mathbf{x}^{(i)})^T(\mathbf{x}^{(i)})$. Therefore, the eigenvalues of the matrix $\mathbf{R}_i^{\frac{T}{2}}\mathbf{W}_{ij}\mathbf{R}_i^{\frac{1}{2}}$ are $\mu_1=\nu_1 \lambda_1=\sigma_n^2 \lambda_1$ and $\mu_2=\nu_K\lambda_2=-(\sigma_n^2+\sigma_h^2 (\mathbf{x}^{(i)})^T(\mathbf{x}^{(i)}))\lambda_1$, which completes the proof. 

\vspace{-0.1cm}
\section{Proof of Proposition \ref{prop1}}\label{prop1_prf}
Following Lemma \ref{lemma1}, the matrix $\mathbf{R}_i^{\frac{T}{2}} \left(\mathbf{R}_i^{-1}-\mathbf{R}_j^{-1}\right)\mathbf{R}_i^{\frac{1}{2}}$ has two non-zero eigenvalues $\mu_1$ and $\mu_2$. Then, from \eqref{pep}, we can write \cite{SKO}
\begin{align}
\mathbb{P}\left\{\sum_{k=1}^K\mu_k |\zeta_k|^2 >\phi_{ij} \right\} = \mathbb{P}\big\{\mu_1 |\zeta_1|^2 \!+\! \mu_2 |\zeta_2|^2 > \phi_{ij} \big\},
\end{align}
where $|\zeta_1|^2$ and $|\zeta_2|^2$ are Gamma random variables with shape parameter $1/2$ and scale parameter $2$, since $\zeta_1, \zeta_2 \sim \mathcal{N}(0,1)$. By Lemma \ref{lemma1}, $\mu_1$ and $\mu_2$ have opposite signs, so solving for $|\zeta_1|^2$ and using both CDF and PDF of the Gamma distribution,
\begin{align}
&\mathbb{P}\{\mathbf{x}^{(i)}\rightarrow \mathbf{x}^{(j)}\} = \frac{1}{\pi\sqrt{2}}\int_{z_0}^\infty \gamma\left(\frac{1}{2}, \frac{\phi_{ij}-\mu_2 z}{2\mu_1}\right) \frac{e^{-\frac{z}{2}}}{\sqrt{z}} dz\nonumber\\
&=\frac{\sqrt{2}}{\pi}\sum_{j=0}^\infty \frac{(-1)^j}{j!(2j+1)} \int_{z_0}^\infty \left(\frac{\phi_{ij}-\mu_2 z}{2\mu_1}\right)^{j+\frac{1}{2}} \frac{e^{-\frac{z}{2}}}{\sqrt{z}} dz,
\end{align}
which follows from the power series representation of the lower incomplete gamma function \cite[8.354]{GRAD} and where $z_0 = \max(0,\phi_{ij}/\mu_2)$. The final result follows with the help of \cite[3.383]{GRAD}.

\vspace{-0.25cm}
\section{Proof of Proposition \ref{prop2}}\label{prop2_prf}
By conditioning on $h$, the average power is given by
\begin{align}
\mathbb{E}\left\{\left|h x^{(N)}[k]+n[k]\right|^2 \Big| h\right\} = h^2 \xi(N) + \sigma^2_n,
\end{align}
where $\xi(N) \triangleq \mathbb{E}\{(x^{(N)}[k])^2\} = \frac{P}{NK}\sum_{k=1}^K\frac{\sin^2[\pi N\Delta f^{(N)} k]}{\sin^2[\pi \Delta f^{(N)} k]}$. Then, the CDF $F_{\rm PAPR}(\theta,N)$ of PAPR is written as
\begin{align}
\mathbb{P}\{\papr(N) < \theta | h\} &= \mathbb{P}\left\{\frac{\max_k |r[k]|^2}{h^2\xi(N) + \sigma^2_n} < \theta \Big| h\right\}.
\end{align}
Since the variables $r[k]$ are mutually independent, the CDF of the largest order statistic is given by
\begin{align}
\prod_k \mathbb{P}\left\{\left|r[k]\right|^2 < \theta\left(h^2\xi(N) + \sigma^2_n\right) \Big| h\right\}.
\end{align}
As $r[k]=h x^{(N)}[k]+n[k]$ and $n\sim \mathcal{N}(0,\sigma^2_n)$, then $\sigma_n r[k] \sim \mathcal{N}(hx^{(N)}[k]/\sigma_n,1)$. Therefore, $|r[k]|^2$ is a non-central chi-squared random variable with one degree of freedom and non-centrality parameter $(hx^{(N)}[k]/\sigma_n)^2$. Then, we can write
\begin{align}
&F_{\rm PAPR}(\theta,N) = \prod_k \mathbb{P}\left\{\left|r[k]\right|^2 < \frac{\theta}{\sigma^2_n}\left(h^2\xi(N) + \sigma^2_n\right) \Big| h\right\}\nonumber\\
&=\prod_k\!\Bigg(\!1\!-\!\mathbb{E}_h\!\left\{\!Q_\frac{1}{2}\!\left(\frac{hx^{(N)}[k]}{\sigma_n},\theta^\frac{1}{2}\!\left(\frac{h^2\xi(N)}{\sigma^2_n} + 1\right)^{\!\frac{1}{2}}\right)\!\right\}\!\!\Bigg),
\end{align}
which follows from the CDF of a non-central random chi-squared variable with one degree of freedom and $Q_\frac{1}{2}(\cdot,\cdot)$ is the Marcum $Q$-function of order $1/2$. By unconditioning on $h$ using the Gaussian distribution PDF completes the proof.

\end{document}